\documentclass[12pt,A4]{article}
\usepackage{latexsym}
\usepackage{amssymb}
\usepackage{amsmath}
\usepackage{mathrsfs}
\usepackage{setspace}
\onehalfspace
\usepackage{fancyhdr}

\usepackage{url} 
\usepackage{latexsym}
\usepackage{amssymb}
\usepackage{amsmath}
\usepackage{amsthm}
\usepackage{graphicx}
\usepackage{mathrsfs}
\usepackage{bigints}
\usepackage{listings}
\theoremstyle{definition}

\newtheorem{theorem}{Theorem}
\newtheorem{lemma}{Lemma}

\newtheorem{corollary}{Corollary}
\newtheorem{algorithm}{Algorithm}

\usepackage{color}

\title{Fine Entanglement and State Manipulation of Two Spin Coupled Qubits: A Lie Theoretic Overview}
\date{$16^{th}$ February 2015}
\author{\it Rod Vance}

\begin{document}

\begin{abstract}
By building on the work in Kuzmak \& Tkachuk, ``Preparation of quantum states of two spin-$\frac{1}{2}$ particles in the form of the Schmidt decomposition'', Physics Letters A, {\bf 378}, pp1469-1474, which outlined the control of the degree of entanglement within this system, it is proven that any $SU(4)$ state manipulation operator can be realised for this system using a sequence of pulsed magnetic fields in either two linearly independent directions if the gyromagnetic ratios are unequal or three directions for equal gyromagnetic ratios. To achieve this goal, an elementary Lie theoretic proof of the fact that the group of transformations generated by finite products of exponentials of a set of Lie algebra vectors is equal to the Lie group generated by the smallest Lie algebra containing those vectors is rewritten into an explicit algorithm. A numerical example as well as the proof of the algorithm's effectiveness is given.
\end{abstract}

\maketitle

\section{Introduction: Physical System Overview}
\label{Introduction}

The physical system considered here is the four dimensional quantum state space of two, two dimensional qubits. Furthermore, in this paper, the two qubits have an unremovable coupling between them, defined by a term proportional to $\sigma_x\,\otimes\,\sigma_x +\sigma_y\,\otimes\,\sigma_y+\sigma_z\,\otimes\,\sigma_z$  as an interaction term in the two-particle system's Hamiltonian. For concreteness I shall consider a lone donor $^{31}P$ atom (the naturally-occuring isotope) in a lattice of silicon as in \cite{KuzmakTkachuk}. This donor atom's nucleus has a spin $\frac{1}{2}$ and at crogenic temperatures the lone  $^{31}P$ in the $Si$ lattice can trap an electron. In this cryogenic, electron-trapped configuration, the  $^{31}P$ donor-electron pair is a coupled-spin, two qubit system. With the quantum state space defined by $\mathcal{H}_n\otimes\mathcal{H}_e$, where $\mathcal{H}_n,\,\mathcal{H}_e$ are, respectively, the separate two-dimensional spin state spaces for the free donor and electron, the Hamiltonian, Schr{\" o}dinger Equation and full coupled spin ket for the coupled system (with the ground state set to energy $0$) in the presence of a magnetic field are:

\begin{equation}
\label{BasicHamiltonian}
\begin{array}{lcl}
\hat{H}&=&\sum\limits_{j\in\{x,\,y,\,z\}}\left(B_j\,\left(\gamma_n\,\sigma_j\otimes\mathrm{id} -\gamma_e\,\mathrm{id}\otimes\sigma_j\right)-\frac{\kappa}{2}\,\sigma_j \otimes \sigma_j\right)\\
\mathrm{d}_t\,\psi&=&e^{i\,\hat{H}\,t}\,\psi;\quad\psi\stackrel{def}{=}\left(\begin{array}{cccc}z_1&z_2&z_3&z_4\end{array}\right)^T\\
|\psi\rangle &\stackrel{def}{=}& z_1\,|\psi_{+1}^n\rangle\otimes|\psi_{+1}^e\rangle +z_2\,|\psi_{+1}^n\rangle\otimes|\psi_{-1}^e\rangle+z_3\,|\psi_{-1}^n\rangle\otimes|\psi_{+1}^e\rangle+z_4\,|\psi_{-1}^n\rangle\otimes|\psi_{-1}^e\rangle
\end{array}
\end{equation}

\noindent where $\sigma_j$ are the $2\times 2$ Pauli spin matrices, $B_j$ are the Cartesian components of the magnetic induction and the subscripts $\pm1$ stand for the spin up/spin down kets in the separate nucleus ($n$) and electron ($e$) spin states. $\gamma_n,\,\gamma_e$ are the nucleus's and electron's gyromagnetic ratio, respectively; these are set by the fundamental physics of these particles and cannot be controlled by the experimenter. Their rough values in SI units, as given in \cite{KuzmakTkachuk}, are $\gamma_n=17.23\,{\rm MHz\,T^{-1}}$ and $\gamma_e=27.97\,{\rm GHz\,T^{-1}}$. Likewise, the value of $\kappa\approx 58.765{\rm MHz}$ is set by the geometry of the overlap between the electon's orbital and the nucleus: it is always present and cannot be experimentally changed. Lastly, there is currently no technology through which an experimenter could impart a different magnetic field separately to the nucleus and the electron; indeed, significant magnetic field variations over the {\aa}ngstrom scales of the system imply $10{\rm keV}$ electromagnetic radiation, which would definitely disturb the system in ways not described by the simple model of \eqref{BasicHamiltonian}. If we make the following definitions of the quaternion units and members of $\mathfrak{su}(4)$:

\begin{equation}
\label{QuaternionDefinitions}
\begin{array}{c}
\mathbf{1} = \left(\begin{array}{cc}1&0\\0&1\end{array}\right);\quad\mathbf{i} = \left(\begin{array}{cc}0&i\\i&0\end{array}\right);\quad\mathbf{j} = \left(\begin{array}{cc}0&-1\\1&0\end{array}\right);\quad\mathbf{k} = \left(\begin{array}{cc}i&0\\0&-i\end{array}\right)\\\\
\begin{array}{ll}
\hat{X}_0= \frac{1}{\gamma_n+\gamma_e}\,\left(-\gamma_n\,\mathbf{i} \otimes \mathbf{1} + \gamma_e\,\mathbf{1}\otimes\mathbf{i}\right);&\hat{Y}_0= \frac{1}{\gamma_n+\gamma_e}\,\left(\gamma_n\,\mathbf{j} \otimes \mathbf{1} - \gamma_e\,\mathbf{1}\otimes\mathbf{j}\right)\\\\\hat{Z}_0= \frac{1}{\gamma_n+\gamma_e}\,\left(-\gamma_n\,\mathbf{k} \otimes \mathbf{1} + \gamma_e\,\mathbf{1}\otimes\mathbf{k}\right);&\hat{K} = \frac{i}{2}\left(\mathbf{i}\otimes\mathbf{i}+\mathbf{j}\otimes\mathbf{j}+\mathbf{k}\otimes\mathbf{k}\right)
\end{array}
\end{array}
\end{equation}

\noindent then the  physical constraints mean that only time evolution operators of the form $e^X$ where $X$ is a Lie algebra member of the form \linebreak$X=\left(B_x\,\hat{X}_0+B_y\,\hat{Y}_0+B_z\,\hat{Z}_0\pm \kappa\,\hat{K}\right)\,\tau$ with the SI unit values cited above are directly available to the experimenter to control the coupled system's state with. Here $\tau$ is the magnetic field pulse time. The $B_j$ can be any real values within an interval set by the magnetic field pulsing apparatus: today roughly $-10{\rm mT}\leq B_j\leq +10{\rm mT}$ is reasonable. The magnitude of the co-efficient of $\hat{K}\,\tau$ is always the same and cannot be varied, being roughly $58.765{\rm MHz}$. However, its sign can be positive or negative, because the one parameter group $\{e^{s\,\hat{K}}|\;s\in\mathbb{R}\}$ is compact and periodic, {\it i.e}  isomorphic to $U(1)$. The eigenvalues of $\hat{K}$ are $\frac{3}{2}\,i$ and $-\frac{1}{2}\,i$ (the latter a triple), so that $e^{K\,t}$ is periodic in $t$ with period $4\,\pi$. Thus, the period of $e^{\kappa\,\hat{K}\,t}$ is $4\,\pi/\kappa\approx 214{\rm ns}$. Therefore a matrix of the form $e^{-\kappa\,\hat{K}\,t_0}$ is of the form $\exp\left(\kappa\,\hat{K}\,\left(\frac{4\,\pi}{\kappa}-t_0\right)\right)$.  $e^{K\,t}$ is interesting from the experimenter's standpoint because it periodically entangles then unentangles any given input state, and thus, as shown in \cite{KuzmakTkachuk}, a sequence of operations of the form $e^{X\,\tau}$ where $X=\left(B_x\,\hat{X}_0+B_y\,\hat{Y}_0+B_z\,\hat{Z}_0\pm \kappa\,\hat{K}\right)\,\tau$ can be used to prepare any quantum state in $\mathcal{H}_n\otimes\mathcal{H}_e$ given the input of one particular state that is readily prepared by the recipe of \cite{Pla}. As another example, the degree of entanglement is defined in \cite{Chen} as the von Neumann entropy of the renormalised projection of the whole state vector in $\mathcal{H}_n\otimes\mathcal{H}_e$ onto either $\mathcal{H}_n$,\,$\mathcal{H}_e$, which projection, considered in its subspace alone, is equivalent to a mixed quantum state. The degree of entanglement can be shown to be a monotonic function of $\mathscr{D}_e=2\,|z_1\,z_4-z_2\,z_3|\in[0,\,1]$, both the entropy and $\mathscr{D}_e$ vary between $0$ (for a product state) and $1$ (for a maximally entangled state) and, for example, $e^{\kappa\,\hat{K}\,t} (\begin{array}{cccc}1&0&0&0\end{array})^T$ swings periodically between a product state and a maximally entangled one. 

In \S\ref{StateControl} I shall use Lie theoretic ideas to find a sequence of  operations of the form $e^{X\,\tau}$ to be equivalent to {\it any} unitary operator in $SU(4)$, even though the subspace of $\mathfrak{su}(4)$ spanned by vectors of the form \linebreak$X=\left(B_x\,\hat{X}_0+B_y\,\hat{Y}_0+B_z\,\hat{Z}_0\pm \kappa\,\hat{K}\right)\, \tau$ is only four dimensional and therefore very ``small'' and ``uncomplicated'' compared to the full Lie algebra $\mathfrak{su}(4)$. $SU(4)$ is a simple Lie group, so one cannot therefore break its description down into a semidirect product of smaller Lie groups, which decomposition might simply the goal of realising any {\it any} unitary operator in $SU(4)$. However, one particular, non-normal subgroup of $SU(4)$ which is of interest is the subgroup of exponentials of the following Lie algebra.

\begin{lemma} 
\label{DegreeOfEntanglementPreserverLemma}
The Lie group $\exp(\mathfrak{d})$ where:


\begin{equation}
\label{DegreeOfEntanglementPreserverLemma_1}
\mathfrak{d}=\mathrm{span}_\mathbb{R}\left(\{i\,\mathbf{1}\otimes\mathbf{1},\,\mathbf{i}\otimes\mathbf{1},\,\mathbf{j}\otimes\mathbf{1},\,\mathbf{k}\otimes\mathbf{1},\,\mathbf{1}\otimes\mathbf{i},\,\mathbf{1}\otimes\mathbf{j},\,\mathbf{1}\otimes\mathbf{k}\}\right)
\end{equation}

\noindent  is precisely the subgroup of $U(4)$ that conserves the degree of entanglement of the states its acts on linearly. That is, $\gamma \in \exp(\mathfrak{d})$ iff $\mathscr{D}_e(\gamma\,\psi) = \mathscr{D}_e(\psi)\,\forall\,\psi\in\mathcal{H}_n\otimes\mathcal{H}_e$ where $\mathscr{D}_e(\psi) = 2\,|z_1\,z_4-z_2\,z_3|$ and $\psi$ and the $z_j$ are defined in \eqref{BasicHamiltonian}.  This Lie group is isomorphic to $U(1)\times SU(2)\times SU(2)$ and comprises all matrices of the form $\boldsymbol{\zeta}_e\otimes\boldsymbol{\zeta}_n$, where $\boldsymbol{\zeta}_e,\,\boldsymbol{\zeta}_n$ are $2\times2$ unitary matrices acting separately on the indivdual spin substates.
\end{lemma}

\begin{proof}
\label{DegreeOfEntanglementPreserver} The precise condition for conservation of $\mathscr{D}_e$ is readily shown to be equivalent to $\gamma=e^H$ where $H$, an  ``infinitessimal $\mathscr{D}_e$ conserver'', fulfills: 

\begin{equation}
\label{DegreeOfEntanglementPreserverLemma_2}
\begin{array}{cl}
&H^T\,Q + Q H = i\,\phi_0\,Q\text{ for some }\phi_0\in\mathbb{R};\;Q=Q^T = \mathbf{j}\otimes\mathbf{j};\\
\Leftrightarrow 
&H=i\,\phi_0\,\mathbf{1}\otimes\mathbf{1}+q_{1\,3}\,\mathbf{i}\otimes\mathbf{1}-p_{1\,3}\,\mathbf{j}\otimes\mathbf{1}+\frac{\beta_1+\beta_2}{2}\,\mathbf{k}\otimes\mathbf{1}+\\&\quad q_{1\,2}\,\mathbf{1}\otimes\mathbf{i}-p_{1\,2}\,\mathbf{1}\otimes\mathbf{j}+\frac{\beta_1-\beta_2}{2}\,\mathbf{1}\otimes\mathbf{k}
;\;\beta_j,\,\phi_0,\,p_{j\,k},\,q_{j\,k}\in\mathbb{R}\end{array}
\end{equation}

\noindent and the claimed isomorphism follows straight away from the definition of $\mathfrak{d}$.\qedhere

\end{proof}

\section{Lie Theoretic Grounding}
\label{LieBackground}

I now turn to the theorem alluded to in the abstract. The full solution of several classic applied mathematics problems rests on this proof; such problems as {\bf ({\it i})} the Nelson Parking Problem\cite{Nelson, RossmannNelsonParkingProblem}, where the $\exp(t_j\,X_j)$ are the changes in a car's configuration (position and orientation on Euclidean $\mathbb{R}^2$) wrought by driving forwards or backwards in a car with a constant path curvature (constant steering setting) and the smallest Lie algebra containing the $X_j$ is the whole of $\mathfrak{e}(2)$ or {\bf ({\it ii})} Montgomery's description\cite{Montgomery} of the torque free falling cat as a fibre bundle with the space of cat shapes $\mathcal{S}$ as the base space, the space of cat's orientation in an inertial frame (described by an element of $SO(3)$) as the fibre and the structure (gauge) group is some Lie subgroup of $SO(3)$, depending on what the smallest Lie algebra generated by the shape shifting moves the cat can make is. The connexion on the bundle arises from computing the shift in the cat's orientation that must accompany a shape shift so as to conserve the cat's total angular momentum. One last problem\cite{Vance}, less noteworthy than these but worth citing here because I personally have formerly dealt with is the proof that a finite system of $N$ planar single mode waveguides can realise any transfer matrix in $U(N)$ even though there is only nearest neighbour coupling between the waveguides. Mathematically, this means that the $X_j$ span only tridiagonal skew-Hermitian $N\times N$ matrices, but the smallest Lie algebra containing these is the whole of $U(N)$. Therefore (since $U(N)$ is connected), a finite sequence of nearest neighbour coupled waveguides can realise any transfer matrix in $U(N)$ by dint of Theorem \ref{GeneratedLieSubgroupTheorem}. A second reason for mentioning this particular application is that, even though Theorem \ref{GeneratedLieSubgroupTheorem} and the methods of the present paper give a solution that can be broadened to any Lie group and algebra, in the case of nearest neighbour coupled waveguides there is a particular solution which, in my opinion, is superior to my own work for some applications and that is the algorithm presented in \cite{ReckZeilinger}. Here a sequence of transfer matrices of $2\times 2$ symmetric coupled waveguides fore-multiply a general $U(N)$ member in a procedure somewhat like Gaussian elimination. The unitarity of the transfer matrix means that rows and columns are reduced to one, unit magnitude element at once (instead of needing reduction separately as in Gaussian elimination), thus $N-1$ basic steps reduces the problem of synthesis of a general element of $U(N)$ to the simpler problem of a realisation of a general element of $U(N-1)$, thence, on application of $N-2$ basic steps, to a member of $U(N-2)$ thus the algorithm, inductively, will decompose the general $U(N)$ to a realisation as a finite sequence of symmetric couplers and phase delays.

I now state and prove in detail the first theorem cited in the abstract; indeed this is a slightly more general result, holding for any set of analytic paths $\sigma:\mathbb{R}\to\mathfrak{G}$ passing through the identity instead of a set of one-parameter groups $\{\exp(t\,X_j)|\;X_j\in\mathfrak{g};\;t\in\mathbb{R}\}$. Results similar to this one seem often to be tacitly assumed in the literature but explicit proofs seem to be few. Many texts justify something like the following theorem with a glib appeal to the Trotter product formulas \linebreak$\exp([X\,Y])=\lim\limits_{n\to\infty} \left(\exp\left(\frac{X}{n}\right)\,\exp\left(\frac{Y}{n}\right)\,\exp\left(-\frac{X}{n}\right)\, \exp\left(-\frac{X}{n}\right)\right)^{n^2}$ and $\exp(X+Y)=\lim\limits_{n\to\infty} \left(\exp\left(\frac{X}{n}\right)\,\exp\left(\frac{Y}{n}\right)\right)^n$. However, an argument grounded on these formulas only establishes a weaker version (stated in \cite{Vance}) of the result in the abstract, to wit, that the {\it closure} of the group generated by the $e^{\tau\,\hat{X}_j}$ is $\exp(\mathfrak{h})$, where $\mathfrak{h}$ is the smallest Lie algebra containing the $\hat{X}_j$, so that a finite product of terms of the form  $e^{\tau\,\hat{X}_j}$ can be arbitrarily near to any element of $\mathfrak{H}=\langle\exp(\mathfrak{h})\rangle$. This is not a restriction or weakening if the smallest Lie algebra concerned is that of a {\it closed} (and topologically embedded) subgroup of the whole, but this of course is not always so. In this latter case, one can then strengthen the weaker result to be equivalent to the result proven below for $GL(N,\,\mathbb{R})$ through the result that every linear Lie group is Lie-isomorphic (if not the same as) to a {\it closed} Lie subgroup of $GL(N,\,\mathbb{R})$ as shown in \cite{Goto}. For example, the irrational slope one parameter subgroup of the 2-torus group $U(1)\times U(1)$ is not closed, but it is isomorphic to $(\mathbb{R},\,+)$ and the latter is a linear Lie group and closed as the subgroup of $GL(2,\,\mathbb{R})$ comprising upper triangular $2\times2$ matrices with ones along the leading diagonal. The following proof avoids all this complexity to get the stronger result directly and constructively; thus can guide the general algorithms of \S\ref{StateControl}. 

\begin{theorem}[Steering a Lie Group Path Without the full rank Tangent Space]
\label{GeneratedLieSubgroupTheorem}

Let $\mathfrak{G}$ be a connected Lie group, $\mathfrak{g}$ its Lie algebra. Kit the appropriate sized nucleus $\mathcal{N}_\mathrm{id}$ with geodesic co-ordinates so that the co-ordinate map is $\lambda = \log:\mathcal{N}_\mathrm{id}\to\mathfrak{g}$ and now let $\sigma_j:[-1,\,1]\to\mathcal{N}_\mathrm{id};\;\sigma_j(0)=\mathrm{id};\,\left.\mathrm{d}_\tau \sigma_j(\tau)\right|_{\tau=0}=\hat{X}_j;\;j=1,\,2,\,\cdots,\,M$ be $M$ $C^\omega$ paths (in the geodesic co-ordinates) through $\mathfrak{G}$ with $\sigma_j(-\tau) = \sigma_j(\tau)^{-1}$. The $\hat{X}_j$ may or may not span the whole of $\mathfrak{g}$; the typical situation, and main point of this theorem, is where they do not.

Let furthermore $\mathfrak{h}$ be the smallest Lie algebra containing the $\hat{X}_j$; otherwise put: $\mathfrak{h}$ is the intersection of all Lie algebras containing the $\hat{X}_j$ and is the set of all entities that can be gotten from the $\hat{X}_j$ by finite sequences of linear (scaling and addition) and Lie bracket operations. Then for every $X\in\mathfrak{h}$ there is a finite number of terms product of the basic paths defined by:

\begin{equation}\label{GeneratedLieSubgroupTheoremProof_1}
\sigma:[-1,\,1]\to\mathfrak{G};\;\sigma(\tau) = \prod\limits_{k=1}^R\,\sigma_k(\alpha_k\,\tau)
\end{equation}

\noindent such that the tangent to the path $\sigma$ at the identity is $\left.\mathrm{d}_\tau \sigma(\tau)\right|_{\tau=0}=X$.

\end{theorem}

\begin{corollary}
\label{GeneratedLieSubgroupTheoremCorollary} 

From Theorem \ref{GeneratedLieSubgroupTheorem}, and from \cite{RossmannSecondKindCanonicalCoordinates} it follows that every member of $\exp(\mathfrak{h})$ and thus every element of the (possibly non topoligically embdedded) Lie subgroup $\mathfrak{H} = \bigcup\limits_{k=1}^\infty\exp(\mathfrak{h})^k$ corresponding to $\mathfrak{h}\subset\mathfrak{g}$ under the Lie Correspondence\cite{RossmannLieCorrespondence} can be realised as a finite product of the form $\prod\limits_{k=1}^Q \sigma_{j(k)}(\tau_k)$ i.e. as a finite product of terms of the form $\sigma_j(\tau_k)$.

\end{corollary}

\begin{proof}[Proof of Theorem \ref{GeneratedLieSubgroupTheorem}]
\label{GeneratedLieSubgroupTheoremProof} 

Given two $C^\omega$ path segments $\sigma_X,\,\sigma_Y:[-1,\,1]\to\mathfrak{G}$ through $\mathfrak{G}$ with $\sigma_X(0)=\sigma_Y(0) = \mathrm{id}$, $\sigma_X(-\tau) = (\sigma_X(\tau))^{-1}$, $\sigma_Y(-\tau) = (\sigma_Y(\tau))^{-1}$ and with tangents $X = \left.\mathrm{d}_\tau\sigma_X(\tau)\right|_{\tau=0}$, $Y= \left.\mathrm{d}_\tau\sigma_Y(\tau)\right|_{\tau=0}$ to the identity, we first show how to realise a path with tangent $[X,\,Y]$ to the identity as a finite product of these paths. We assume, of course, that $[X,\,Y]$ must be nonzero and linearly independent of $X$ and $Y$ (otherwise there is nothing to prove: the smallest Lie algebra containing $X$ and $Y$ is the vector space spanned by them).

Thus let:

\begin{equation}
\label{GeneratedLieSubgroupTheoremProof_2}
\begin{array}{lcl}
\sigma_X(\tau) &=& \exp\left(\tau\,X + \sum\limits_{j=1}^N p_j(\tau)\,\hat{X}_j\right)\\
\sigma_Y(\tau) &=& \exp\left(\tau\,Y + \sum\limits_{j=1}^N q_j(\tau)\,\hat{X}_j\right)
\end{array}
\end{equation}

\noindent where the $p_j(\tau)$ and $q_j(\tau)$ are analytic functions of $\tau$ comprising only second and higher powers of $\tau$ and we must allow for the possibility of all the $\hat{X}_j$ that span $\mathfrak{g}$ being present as terms multiplied by second and higher powers of $\tau$. So now we consider the tangent at the identity to the family of $C^\omega$ paths parameterised by the parameter $s$ as follows:

\begin{equation}
\label{GeneratedLieSubgroupTheoremProof_3}
\sigma_s:[-1,\,1]\to\mathfrak{G};\quad\sigma_s(\tau) = \sigma_X(s)\,\sigma_Y(\tau)\,\sigma_X(-s)
\end{equation}

\noindent so that the tangent to a family member at the identity as a function of the parameter $s$ is, by elementary calculation:

\begin{equation}\label{GeneratedLieSubgroupTheoremProof_4}
\begin{array}{lcl}T(s) &=& \exp\left(\mathrm{ad}(X)\,s + \sum\limits_{j=1}^N p_j(s)\,\mathrm{ad}(\hat{X}_j)\right)\, Y \\ &=& (1+\tilde{p}_0(s))\,Y + (s+\tilde{p}_1(s))\,\mathrm{ad}(X)\,Y + \sum\limits_{j=2}^{N-1} \tilde{p}_j(s)\,\hat{Z}_j\end{array}
\end{equation}

\noindent where we have absorbed all the higher powers $s^2 \mathrm{ad}(X)^2\,Y/2!,\,s^3 \mathrm{ad}(X)^3\,Y/3!,\,\cdots$ into the sum $\sum\limits_{j=2}^{N-1} \tilde{p}_j(s)\,\hat{Z}_j$ with modified analytic co-efficient functions $\tilde{p}_j(s)$. Here the $\hat{Z}_j\in\mathfrak{g}$ are vectors which are mutually linearly independent and also linearly independent from either $Y$, $[X,\,Y]$ or any linear combination of these last two. We have allowed for the possibility of the vector $\mathrm{ad}(X)\,Y$ showing up in the high order powers of $s$ by adding $\tilde{p}_1(s)$ to the multplier $s$ of the term $\mathrm{ad}(X)\,Y$; likewise the $\tilde{p}_0(s)$ term absorbs any $Y$ component from high order powers. The co-efficient functions $\tilde{p}_j(s)$, including different powers of $s$, are all different and so are linearly independent on any interval. Therefore, we can choose $m$ discrete $s_j\in[-1,\,1]$ to yield $m$ linearly independent tangent vectors $T(s_j)$ where $m\leq N$ is the total number of linearly independent vectors in the set $\{Y,\,[X,\,Y],\,\mathrm{ad}(X)^2\,Y,\,\mathrm{ad}(X)^3\,Y,\,\cdots\}$. Therefore, through Gaussian elimination, we can linearly combine the tangent vectors $T(s_j)$ so that their linear combination equals the term $[X,\,Y]$. That is, we can find superpositions weights $\alpha_j$ such that:

\begin{equation}\label{GeneratedLieSubgroupTheoremProof_5}[X,\,Y] = \sum\limits_{k=1}^m\alpha_k\,T(s_k)\end{equation}

Therefore, since the Lie group's identity tangent space is a linear space, the path $\sigma(\tau) = \prod\limits_{k=1}^m\sigma_X(s_k)\,\sigma_Y(\alpha_k\,\tau)\,\sigma_X(-s_k)$ has the tangent $[X,\,Y]$ to the identity. Note that in most cases, not all the vectors spanning the Lie algebra $\mathfrak{g}$ are present in the sum in $\eqref{GeneratedLieSubgroupTheoremProof_4}$, so that the sum in $\eqref{GeneratedLieSubgroupTheoremProof_5}$ contains fewer ($m$) than $N$, terms. It may even be that none of these terms are there aside from $Y$ and $[X,\,Y]$, but the point is that these last two are guaranteed to be in $\eqref{GeneratedLieSubgroupTheoremProof_4}$ with the nonzero weights shown there, so that the procedure above can be summarised as: $T(s)$ includes $Y, \,[X,\,Y]$ and some other linear independent vectors in the Lie algebra with linearly independent co-efficient functions of $s$. Therefore, there exists a linear combination of the some $T(s_k)$ that sums to any of these linearly independent vectors, thus, in particular we can retrieve $Y$ or $[X,\,Y]$ in this way.

Given we now have a set of $C^\omega$ paths with tangents $X_j$, we can build paths with a tangent to the identity equal to any linear superpositon of the $X_j$ by the method of \cite{RossmannSecondKindCanonicalCoordinates, RossmannLieCorrespondence}, and we can build a path with a tangent to the identity equal to the Lie bracket of any pair of these paths with the procedure described above. Therefore a finite sequence of operations comprising the above procedure together with linear superposition operations by the method of \cite{RossmannSecondKindCanonicalCoordinates} can realise a $C^\omega$ path with any tangent inside some neighbourhood $\mathcal{H}$ of $\mathbf{0}$ in the smallest Lie algebra $\mathfrak{h}$ containing the $X_j$.

Having built paths whose tangents span $\mathfrak{h}$, we can, arguing as in \cite{RossmannSecondKindCanonicalCoordinates},realise any member of $\exp(\mathcal{H})$ on a finite product of these spanning paths. From there we can see that a finite product of such paths can realise any member of the connected Lie group $\mathfrak{H} = \bigcup\limits_{k=1}^\infty\exp(\mathcal{H})^k = \bigcup\limits_{k=1}^\infty \exp(\mathfrak{h})^k$. \qedhere

\end{proof}

Sifting carefully through this proof, one can specialize it to the paths $\sigma_j(\tau)=\exp(X_j)$ and at the same time reword it to make it almost explicitly constructive.

\begin{algorithm}
\label{LieAlgebraBasisFindingAlgorithm}

Let $\mathfrak{G}$, $N$, $\mathfrak{g}$ and $\mathcal{N}_\mathrm{id}$ be as in Theorem \ref{GeneratedLieSubgroupTheorem} and now specialise $\{\sigma_j:[-1,\,1]\to\mathcal{N}_\mathrm{id};\;\sigma_j(\tau)=\exp(\tau\,\hat{X}_j)\}_{j=1}^M$ where $\{\hat{X}_j\}_{j=1}^M$ are linearly independent and let $\mathfrak{h}$ bethe smallest Lie algebra containing the $\hat{X}_j$; aside from the specialisation to $\exp$, {\it all} definitions are wholly analogous to Theorem \ref{GeneratedLieSubgroupTheorem}. Then the following algorithm terminates, and it does so precisely when it has found a basis for the whole of $\mathfrak{h}$. 

\begin{equation*}
\begin{array}{l}
\mathfrak{o} := \{\hat{X}_j\}_{j=1}^M;\quad\text{lengthO} := M;\quad\text{done} := \text{False};\quad \mathcal{S}=\emptyset;\\
\text{\bf while not }\text{done }\text{\bf do}\;\{\\
\quad \text{\bf for each }j\in\{1,\cdots\,\text{lengthO}\}\,\text{\bf do}\;\{\\
\quad\quad\text{\bf for each }k\in\{j,\,\cdots,\,\text{lengthO}\}\,\text{\bf do}\;\{\\
\begin{array}{lcl}
\quad\quad\quad X&:=&\mathfrak{o}_j;\quad Y:=\mathfrak{o}_k;\\
\quad\quad\quad\mathcal{U}&:=&\text{maximal linearly independent set from }Y,\,\mathrm{ad}(X)\,Y,\,\mathrm{ad}(X)^2\,Y,\cdots;\\
\quad\quad\quad n&:=&\mathrm{length}(\mathcal{U});\\
\quad\quad\quad \mathcal{S}&:=&\{s_1,\,s_2,\,\cdots,\,s_n\}\subset\mathbb{R}\,\ni\,e^{s_i\,\mathrm{ad}(X)}\,Y\text{ are linearly independent};\\
\end{array}\\
\quad\quad \}\\
\quad\quad \mathcal{A} := \{e^{s_i\,\mathrm{ad}(X)}\,Y|\;s_i\in\mathcal{S}\};\\
\quad\quad \text{Cull any member of set } \mathcal{A}\text{ which is linearly dependent with set }\mathfrak{o};\\
\quad\quad \text{Append culled }\mathcal{A}\text{ to set }\mathfrak{o};\\
\quad \}\\
\quad \text{\bf if }\text{length}(\mathfrak{o}) =\text{lengthO \,\text{\bf then }done := \text{True};\,\text{\bf else }\text{lengthO} := length}(\mathfrak{o});\\
\}\\
\text{\bf return }\mathfrak{o};
\end{array}
\end{equation*}

\noindent Here the set of vectors being built is the variable $\mathfrak{o}$.

\end{algorithm}

\begin{proof}[Proof that Algorithm \ref{LieAlgebraBasisFindingAlgorithm} terminates and builds the whole of $\mathfrak{h}$]
\label{LieAlgebraBasisFindingAlgorithmProof} The only not-explicitly-constructive part of this algorithm is the finding of \linebreak$\mathcal{S}:=\{s_1,\,s_2,\,\cdots,\,s_n\}\subset\mathbb{R}\,\ni\,e^{s_i\,\mathrm{ad}(X)}\,Y$ are linearly independent inside the innermost for loop. But since $\exp(s\,\mathrm{ad}(X)\,Y)$ is a linear combination of members of:

\begin{equation}
\label{LieAlgebraBasisFindingAlgorithm_1}
\mathcal{U}=\bigcup\limits_{k=0}^N\{\mathrm{ad}(X)^k\,Y\}=\{Y,\,\mathrm{ad}(X)\,Y,\,\mathrm{ad}(X)^2\,Y,\,\mathrm{ad}(X)^3\,Y,\,\cdots\}
\end{equation}

\noindent the proof of Theorem \ref{GeneratedLieSubgroupTheoremProof} straight after \eqref{GeneratedLieSubgroupTheoremProof_5} shows that {\bf ({\it i})} we can find such a set by Gaussian elimination and {\bf ({\it ii})} that vector space spanned by this set is precisely the same as the vector space spanned by $\mathcal{U}$. So by dint of this existence, we can assume that we have some algorithm that implements this one step.

Now witness that the inner for loops add only vectors linearly independent from each other and linearly independent from all members of the set $\mathfrak{o}$. Since there are at most only finitely many, $N$ of these vectors, the algorithm clearly terminates. Hence $\mathrm{span}(\mathfrak{o})\subseteq\mathfrak{h}$ at all times.

Now consider the $\text{while}$ loop straight before its very last pass. In this pass, as just shown above, it calculates sets of vectors that span every set of the form $\bigcup\limits_{k=0}^N\{\mathrm{ad}(X)^k\,Y\}$ for every pair $X\,Y\in\mathfrak{o}$; in particular is calculates a vector space containing $[X,\,Y]$ for every pair $X,\,Y\in\mathfrak{o}$. Since the algorithm terminates, this last pass cannot add new vectors to $\mathfrak{o}$. Hence $\mathfrak{h}\subseteq\mathrm{span}(\mathfrak{o})$ when the algorithm terminates.\qedhere
\end{proof}

I have shown the existence of appropriate $\mathcal{S}$ at each step within the innermost for loop, so for the purposes of this paper I shall simply assume that we have some ``black box" that can reckon appropriate $s_j$ for us. Simple numerical conjugate gradient optimisation will be used in \S\ref{StateControl}.  The possible choices are highly non-unique: almost any randomly chosen different real numbers will work. However, the choice of the $s_j$ does bear strongly on ``how linearly independent'' the calculated basis for $\mathfrak{h}$ turns out to be, {\it i.e.} it bears on the condition number of the array of linearly independent vectors. This condition number in turn will bear on how many basic magnetic field pulse sequences will be needed in \S\ref{StateControl} and with what pulse times. Poor condition numbers mean long and complicated pulse sequences comprising long lingering pulses, and therefore this means that in practice the pulse sequence is likely to be highly sensitive to errors in the experimentally applied magnetic fields or pulse times. In some cases, one may be able to use the Cartan decomposition to find a generalised ``polar'' decomposition for $SU(4)$ thus achieving a better or minimal sequence of evolution operations as done in \cite{Khaneja}, but this last method's applicability depends on the Cartan pair's Lie subalgebra $\mathfrak{k}\subseteq \mathfrak{g}$ and its orthogonal (with respect to the Killing form on $\mathfrak{g}$) complement  $\mathfrak{k}^\perp$ being directly available, as operators from $\exp(\mathfrak{k})$ and $\exp(\mathfrak{k}^\perp)$ to the experimenter as time evolution operators; in general, however, at least one of the available vectors of the form found in Algorithm \ref{LieAlgebraBasisFindingAlgorithm} will ``straddle'' the Cartan pair (have nonzero components in both $\mathfrak{k}$ and $\mathfrak{k}^\perp$). The finding and optimisation of a systematic algorithm, grounded on the Cartan decomposition or otherwise, for calculating $\mathcal{S}$ is therefore a high priority for future work on the procedure of \S\ref{StateControl}.

Once we  have constructed a full basis $\{\hat{Z}_j\}_{j=1}^{\dim\,\mathfrak{h}}$ for $\mathfrak{h}$ with Algorithm \ref{LieAlgebraBasisFindingAlgorithm}, and given any $\gamma\in\mathfrak{H}=\left<\exp(\mathfrak{h})\right>$ within the connected Lie group $\mathfrak{H}=\left<\exp(\mathfrak{h})\right>=\bigcup\limits_{k=1}^\infty\,\exp(\mathfrak{h})^k$ can then, by Corollary \ref{GeneratedLieSubgroupTheoremCorollary}, find real numbers $\tau_{j,\,k}$ so that $\gamma=\prod\limits_{j=1}^{j_0}\,\left(\prod\limits_{k=1}^{\dim\,\mathfrak{h}}\,\exp\left(\tau_{j\,k}\,\hat{Z}_k\right)\right)$ for some finite $j_0$. There are two more steps to this construction. Firstly, we need to decompose $\gamma$ as $\gamma = \prod_{j=1}^{j_1}\,e^{Y_j}$ for $Y_j\in\mathfrak{h}$. If the Lie group $\mathfrak{H}$ is compact, for example, then it is the exponential of its own Lie algebra we can {\it always} find a (nonunique) $Y\in\mathfrak{h}$ such that $\gamma=e^Y$. But a general, connected Lie group ({\it e.g.} $SL(2,\,\mathbb{C})$ or $SL(2,\,\mathbb{R})$) is strictly bigger than simply the exponential of its Lie algebra: there are some elements which are not exponentials of the algebra. In \S\ref{StateControl} we deal with the compact group $SU(4)$, so this first step is not needed and will not be further considered in this paper, but in general a decomposition algorithm must be found. For example, in a connected, semisimple Lie group, the Iwasawa decomposition\cite{Knapp} will achieve this goal.

Thus I now assume that we must realise $\gamma = e^Y$ where $Y\in\mathfrak{h}$. Since one can write $\gamma = e^Y = \exp\left(\frac{Y}{n}\right)^n$ for any $n\in\mathbb{N}$, there is an $n$ big enough that $\exp\left(\frac{Y}{n}\right)\in\mathcal{K}$, where $\mathcal{K}$ is any arbitrarily small nucleus. We choose the latter's size right and then the Wei-Norman equations\cite{WeiNorman} can be used to calculate the canonical co-ordinates as follows:

\begin{lemma}[Wei-Norman]
\label{WieNormanLemma}
Let $\mathfrak{G}$ be an $N$ dimensional connected Lie group, $\mathfrak{g}$ its Lie algebra,  $\{\hat{X}_j\}_{j=1}^N\subset\mathfrak{g}$ vectors spanning $\mathfrak{g}$ and $\mathcal{K}\subset\mathfrak{G}$ a nucleus small enough to be labelled by the canonical co-ordinates of the second kind, {\it i.e.} $\mathcal{K}\subseteq \left\{\prod\limits_{j=1}^N\,e^{\tau_j\,\hat{X}_j}|\;\tau_j\in [-1,\,1]\right\}$. Let $\gamma=e^X\in\mathcal{K}$, then the canonical co-ordinates $\{\tau_j\}_{j=1}^N$ of the second kind for $\gamma=e^X$ where $X = x_1\,\hat{X}_1+x_2\,\hat{X}_2+\cdots$ are calculated by solving the $N$-dimensional nonlinear differential equation:

\begin{equation}
\label{WieNormanLemma_1}
\begin{array}{l}
\mathrm{d}_t\boldsymbol{\tau}=\mathbf{M}(\tau_1,\,\tau_2,\,\cdots)^{-1} \left(\begin{array}{c}x_1\\x_2\\x_3\\\vdots\end{array}\right)\\
\mathbf{M}(\tau_1,\,\tau_2,\,\cdots)=\\\\\quad\left(\begin{array}{c}1\\0\\0\\\vdots\end{array}\,\left(e^{\tau_1\,\mathrm{ad}(\hat{X}_1)}\right)_2\,\left(e^{\tau_1\,\mathrm{ad}(\hat{X}_1)}\,e^{\tau_2\,\mathrm{ad}(\hat{X}_2)}\right)_3\,\cdots\,\left(e^{\tau_1\,\mathrm{ad}(\hat{X}_1)}\,\cdots\,e^{\tau_{N-1}\,\mathrm{ad}(\hat{X}_{N-1})}\right)_N\right)
\end{array}
\end{equation}

\noindent subject to the initial conditions $\boldsymbol{\tau}=\mathbf{0}$ at $t=0$. The required canonical co-ordinates are the values in the vector $\boldsymbol{\tau}$ at $t=1$, {\it i.e.} $X=\prod\limits_{j=1}^N\,e^{\tau_j(1)\,\hat{X}_j}$.

\end{lemma}

\begin{proof}
\label{WieNormanLemmaProof}
See \cite{WeiNorman} for details. The essential trick is to break the derivative \linebreak$\mathrm{d}_t\,e^{\tau_1(t)\,\hat{X}_1}\,e^{\tau_2(t)\,\hat{X}_2}\cdots$ up into its summands by the product rule, and then to rewrite each term of the form $e^{\tau_1(t)\,\hat{X}_1}\,e^{\tau_2(t)\,\hat{X}_2}\,\cdots \frac{\mathrm{d}}{\mathrm{d}\,t}\left(e^{\tau_m(t)\,\hat{X}_m}\right)\,\cdots e^{\tau_N(t)\,\hat{X}_N}$ as \linebreak$e^{\tau_1\,\mathrm{ad}(\hat{X}_1)}\,e^{\tau_2\,\mathrm{ad}(\hat{X}_2)}\cdots e^{\tau_2\,\mathrm{ad}(\hat{X}_{m-1})} \frac{\mathrm{d}}{\mathrm{d}\,t} \tau_m(t)\,\hat{X}_m\,e^{\tau_1(t)\,\hat{X}_1}\,e^{\tau_2(t)\,\hat{X}_2}\cdots e^{\tau_2(t)\,\hat{X}_N}$ by the Wei-Norman``shuffle'' trick. One then uses this to write the differential equation describing the path $[-1,\,1]\to \mathcal{K};\; t\mapsto e^{t\,X}$ in the form \linebreak$X\,e^{t\,X} = \left(\dot{\tau}_1\,X_1 +\dot{\tau}_2\, e^{\tau_1\,\mathrm{ad}(\hat{X}_1)}\,X_2 + \dot{\tau}_2\, e^{\tau_1\,\mathrm{ad}(\hat{X}_1)}\,e^{\tau_2\,\mathrm{ad}(\hat{X}_2)}\,X_3\,\cdots\right)e^{t\,X}$, cancels the $e^{t\,X}$ terms and inverts the matrix to get \eqref{WieNormanLemma_1}. The matrix to be inverted in \eqref{WieNormanLemma_1} is such that its $j^{th}$ column is the $j^{th}$ column of the square matrix \linebreak$e^{\tau_1\,\mathrm{ad}(\hat{X}_1)}\,e^{\tau_2\,\mathrm{ad}(\hat{X}_2)}\cdots\,e^{\tau_{j-1}\,\mathrm{ad}(\hat{X}_{j-1})}$ (when the $e^{\tau_k\,\mathrm{ad}(\hat{X}_k)}$ are written as matrices with respect to the basis $\{\hat{X}_j\}_{j=1}^N$). The matrix to be inverted is the identity matrix at $\boldsymbol{\tau}=\mathbf{0}$ and its determinant is a continuous (indeed analytic) function of the $\tau_j$, therefore there is a nonzero interval $-\epsilon<t<+\epsilon$ for $\epsilon>0$ wherein the matrix is invertible and so \eqref{WieNormanLemma_1} is a Cauchy initial value problem fulfilling the conditions of the Picard-Lindel{\"o}f theorem, {\it i.e.} that the vector function on the right of \eqref{WieNormanLemma_1} is Lipschitz continuous function of the $\tau_j$. Therefore, as long as $X\in\mathfrak{g}$ is near enough to $\mathbf{0}$, the unique, existing solution of the Cauchy initial value problem defined by \eqref{WieNormanLemma_1} must define the co-ordinates for $X$ at $t=1$.\qedhere

\end{proof}

\noindent and this lemma yields the algorithm to realise any $\gamma\in\mathfrak{G}$:

\begin{algorithm}
\label{CanonicalCoordinatesFindingAlgorithm}

Let $\mathfrak{G}$, $N$ and $\mathfrak{g}$ be as in Theorem \ref{GeneratedLieSubgroupTheorem}, $\mathfrak{g}$ be spanned by independent vectors $\{\hat{X}_j\}_{j=1}^N$  and suppose further that $\mathfrak{G}=\exp(\mathfrak{g})$, {\it i.e.} that every group member can be written as an exponential of a Lie algebra member. Let $\gamma=e^X;\,X\in\mathfrak{g}$ be an arbitrary member of $\mathfrak{G}$. Then the following algorithm:

\begin{equation*}
\begin{array}{l}
n:=1;\quad  \det := 1;\quad t:=0;\quad\boldsymbol{\tau}=\mathbf{0};\\
Y:=X;\\
\text{\bf for } t\in[0,\,1]\, \text{\bf do } \{\\
\quad \text{Calculate }\boldsymbol{\tau}(t)\text{ by Wei-Norman equations for } Y;\\
\quad \det=\det(\mathbf{W}^{-1});\\
\quad \text{\bf if }\det<\epsilon\, {\bf break};\\
\}\\
\text{\bf if }\det \geq\epsilon\,\text{\bf then }{\bf return}\;\boldsymbol{\tau},\;1;\; \text{\bf else }\\
\{\\
\quad n:=\mathrm{Floor}\left(\frac{1}{t}\right);\\
\quad Y:=\frac{X}{n};\\
\}\\
\text{\bf for } t\in[0,\,1]\, \text{\bf do } \{\\
\quad \text{Calculate }\boldsymbol{\tau}(t)\text{ by Wei-Norman equations for } Y;\\
\}\\
\text{\bf return}\;\boldsymbol{\tau},\;n
\end{array}
\end{equation*}

\noindent terminates and outputs the canonical co-ordinates for $\gamma_n=\exp\left(\frac{X}{n}\right)$ in $\boldsymbol{\tau}$ and the power $n$ that one must raise $\gamma_n$ to to realise the original $\gamma=e^X$. 

\end{algorithm}

\begin{proof}[Proof of efficacy and convergence of Algorithm] Clear from Lemma \ref{WieNormanLemma}.\end{proof}

The for loop is shown over a continuous iterator variable $t$, so that, as written, the loop must run $\aleph_1$ times! In practice, the loop is a discrete numerical solving of the differential equation and it breaks, unsuccessfully, if the determinant of the matrix drops below a critical threshold value $\epsilon$.  Since the Wei Norman equations will succeed (barring some numerical difficulty) in finding co-ordinates for $\gamma_n$ near enough to the identity, either the first for loop (with $n=1$) will return the correct co-ordinates or the second for loop will do so for $\frac{X}{n}$ with $n$ now calculated to be high enough to bring $\exp\left(\frac{X}{n}\right)$ near enough to the identity to belong to the nucleus $\mathcal{K}$ of the statement of Lemma \ref{WieNormanLemma}.

This ends our Lie theoretical background discussion and I now turn to the synthesis of time evolution operators in $SU(4)$ for coupled two qubit quantum systems.

\section{Coupled Spin State Manipulation}
\label{StateControl}

To further analyse the control that the experimenter wields in applying sequences of magnetic field pulses to the coupled spin system of \S\ref{Introduction}, I define the following $4\times4$ traceless, skew-Hermitian matrices:

\begin{equation}
\label{FullLieAlgebraDefinition}
\begin{array}{l}
\begin{array}{lll}
\hat{X} = \hat{X}_0+\kappa\,\hat{K}&\hat{Y} = \hat{Y}_0+\kappa\,\hat{K}&\hat{Z} = \hat{Z}_0+\kappa\,\hat{K}\end{array}\\
\begin{array}{ll}
\hat{K}_X = \frac{i}{2}\left(-\mathbf{j}\otimes\mathbf{k}+\mathbf{k}\otimes\mathbf{j}\right)&\hat{K}_Y = \frac{i}{2}\left(\mathbf{k}\otimes\mathbf{i}-\mathbf{i}\otimes\mathbf{k}\right)\\
\hat{K}_Z = \frac{i}{2}\left(-\mathbf{i}\otimes\mathbf{j}+\mathbf{j}\otimes\mathbf{i}\right)&
\hat{K}_{XX} =  -\frac{i}{2}\left(\mathbf{j}\otimes\mathbf{j}+\mathbf{k}\otimes\mathbf{k}\right)\\
\hat{K}_{YY} =  -\frac{i}{2}\left(\mathbf{i}\otimes\mathbf{i}+\mathbf{k}\otimes\mathbf{k}\right)&
\hat{K}_{ZZ} =  -\frac{i}{2}\left(\mathbf{i}\otimes\mathbf{i}+\mathbf{j}\otimes\mathbf{j}\right)\\
\hat{L}_X=-\frac{1}{4}\left(\mathbf{i}\otimes\mathbf{1}+\mathbf{1}\otimes\mathbf{i}\right)&\hat{L}_Y=\frac{1}{4}\left(\mathbf{j}\otimes\mathbf{1}+\mathbf{1}\otimes\mathbf{j}\right)\\\hat{L}_Z=-\frac{1}{4}\left(\mathbf{k}\otimes\mathbf{1}+\mathbf{1}\otimes\mathbf{k}\right)&
\end{array}\\
\hat{K}_{XY}=-\frac{i}{2\left(\gamma_n+\gamma_e\right)}\,\left(\gamma_n\,\mathbf{i}\otimes\mathbf{j}+\gamma_e\,\mathbf{j}\otimes\mathbf{i}\right)-\frac{\kappa}{2}\left(\mathbf{i}\otimes\mathbf{1}-\mathbf{1}\otimes\mathbf{i}\right)\\
\hat{K}_{YZ}=-\frac{i}{2\left(\gamma_n+\gamma_e\right)}\,\left(\gamma_n\,\mathbf{j}\otimes\mathbf{k}+\gamma_e\,\mathbf{k}\otimes\mathbf{j}\right)-\frac{\kappa}{2}\left(-\mathbf{j}\otimes\mathbf{1}+\mathbf{1}\otimes\mathbf{j}\right)\\\hat{K}_{ZZ}=\frac{i}{2\left(\gamma_n+\gamma_e\right)}\,\left(\gamma_n\,\mathbf{k}\otimes\mathbf{i}+\gamma_e\,\mathbf{i}\otimes\mathbf{k}\right)-\frac{\kappa}{2}\left(\mathbf{k}\otimes\mathbf{1}-\mathbf{1}\otimes\mathbf{k}\right)\\
\begin{array}{ll}\hat{X}_{KK}=\frac{1}{2}\left(\mathbf{i}\otimes\mathbf{1}-\mathbf{1}\otimes\mathbf{i}\right)&\hat{Y}_{KK}=\frac{1}{2}\left(-\mathbf{j}\otimes\mathbf{1}+\mathbf{1}\otimes\mathbf{j}\right)\\\hat{Z}_{KK}=\frac{1}{2}\left(\mathbf{k}\otimes\mathbf{1}-\mathbf{1}\otimes\mathbf{k}\right)&
\end{array}
\end{array}
\end{equation}

\begin{lemma}
\label{FullLieAlgebraLemma}Given the definitions of \eqref{QuaternionDefinitions} and \eqref{FullLieAlgebraDefinition} and the fixed, strictly positive real number $\kappa$ we have:

\noindent {\bf(a)} The smallest Lie algebra containing  $\{B_x\,\hat{X}+B_y\,\hat{Y}+B_z\,\hat{Z}\pm \hat{K}|\;B_j\in\mathbb{R}\}$ is the whole of $\mathfrak{su}(4)$ whenever at least one of the positive constants $\gamma_n,\,\gamma_e$ are nonzero;  

\noindent {\bf(b)} Suppose that at least one of the positive constants $\gamma_n,\,\gamma_e$ are nonzero and that $\gamma_n\neq \gamma_e$. Let $X(\theta,\,\phi,\,\alpha)=\sin\theta\,\cos\phi\,\hat{X}+\sin\theta\,\sin\phi\,\,\hat{Y}+\cos\theta\,\hat{Z}+\alpha\,\hat{K}$. Then the smallest Lie algebra containing  $\mathbf{b}(\theta,\,\phi)\stackrel{def}{=}\{X(\theta,\,\phi,\,\alpha)|\;\alpha\in\mathbb{R}\}$ (for fixed $\theta,\,\phi\in\mathbb{R}$) is the five dimensional Lie algebra \linebreak$\mathfrak{b}(\theta,\,\phi)=\mathrm{span}_\mathbb{R}\left(\{X_{\theta,\,\phi},\,K,\,\mathrm{ad}(X_{\theta,\,\phi})\,K,\,\mathrm{ad}(X_{\theta,\,\phi})^2\,K,\,\mathrm{ad}(K)^2\,X_{\theta,\,\phi}\}\right)$
(where we write $X_{\theta,\,\phi}= X(\theta,\,\phi,\,0)$) and $\exp(\mathfrak{b}(\theta,\,\phi))$ is a Lie group isomorphic either to $SU(2)\rtimes (U(1)\times \mathbb{R})$ or $SU(2)\rtimes (U(1)\times U(1))\cong SU(2)\rtimes \mathbb{T}_2$.

\noindent {\bf(c)}If $\gamma_n=\gamma_e\neq0$, then the Lie algebra $\mathfrak{b}(\theta,\,\phi)$ in statement {\bf(b)} shrinks to the four dimensional $\mathfrak{b}(\theta,\,\phi)=\mathrm{span}_\mathbb{R}\left(\{X_{\theta,\,\phi},\,K,\,\mathrm{ad}(X_{\theta,\,\phi})\,K,\,\mathrm{ad}(X_{\theta,\,\phi})^2\,K\}\right)$ and $\exp(\mathfrak{b}(\theta,\,\phi))$ is a Lie group isomorphic to $SU(2)\rtimes U(1)$;

\noindent {\bf(d)}If at least one of $\gamma_n,\,\gamma_e$ are nonzero and $\gamma_n\neq\gamma_e$, then the smallest Lie algebra containing both $\mathbf{b}(\theta_1,\,\phi_1)$ and $\mathbf{b}(\theta_2,\,\phi_2)$ for at least one of $\phi_1\neq\phi_2$, $\theta_1\neq\theta_2$ is the whole of  $\mathfrak{su}(4)$ and its exponential is the whole of $SU(4)$;

\noindent {\bf(e)}If $\gamma_n=\gamma_e\neq0$ then the smallest Lie algebra  containing both $\mathbf{b}(\theta_1,\,\phi_1)$ and $\mathbf{b}(\theta_2,\,\phi_2)$ for at least one of $\phi_1\neq\phi_2$ is no longer the whole of $\mathfrak{su}(4)$ but is instead a nine dimensional Lie algebra with $\mathfrak{su}(3)$ as an ideal and its exponential is $SU(3)\rtimes U(1)$ (naturally with $SU(3)$ a normal subgroup).

\end{lemma}

\begin{proof}
\label{FullLieAlgebraLemmaProof}
{\bf(a)} The fifteen matrices $\hat{X}$, $\hat{Y}$, $\hat{Z}$, $\hat{K}_X$, $\hat{K}_Y $, $\hat{K}_Z $, $\hat{K}_{XX} $, $\hat{K}_{YY} $, $\hat{K}_{ZZ} $, $\hat{L}_X $, $\hat{L}_Y $, $\hat{L}_Z $, $\hat{K}_{XY} $, $\hat{K}_{YZ} $, $\hat{K}_{ZX}$ are readily shown to be unconditionally linearly independent. Indeed if they are vectorised into their 16 real-component superposition weights over the basis $\{\mathbf{u}_j\otimes \mathbf{u}_k|\;j,\,k\in 0\cdots 3\}$ and $\mathbf{u}_0=i\,\mathbf{1},\mathbf{u}_1=\mathbf{i},\,\mathbf{u}_2=\mathbf{j},\,\mathbf{u}_3=\mathbf{k}$ and built, together with the same vectorisation of $i$ times the $4\times 4$ identity, into a square matrix of 16, 16-component row-vectors, then the determinant of this array is always $-\frac{1}{16}$ as long as at least one of $\gamma_n,\,\gamma_e$ are nonzero. These fifteen are derived from the algebra members $\hat{X},\,\hat{Y},\,\hat{Z},\,\hat{K}$ by the bracket relationships: 

\begin{equation}
\label{FullLieAlgebraBrackets}
\begin{array}{ccc}
\hat{K}_X=\frac{1}{2}\left[\hat{K},\,\hat{X}\right]&\hat{X}_{KK}=\frac{1}{2}\left[\hat{K},\,\hat{K}_X\right]&\hat{K}_{XX}=\frac{1}{2}\left[\hat{K}_X,\,\hat{X}\right]+\kappa\,\hat{X}_{KK}\\
\hat{K}_Y=\frac{1}{2}\left[\hat{K},\,\hat{Y}\right]&\hat{Y}_{KK}=\frac{1}{2}\left[\hat{K},\,\hat{K}_Y\right]&\hat{K}_{YY}=\frac{1}{2}\left[\hat{K}_Y,\,\hat{Y}\right]+\kappa\,\hat{Y}_{KK}\\
\hat{K}_Z=\frac{1}{2}\left[\hat{K},\,\hat{Z}\right]&\hat{Z}_{KK}=\frac{1}{2}\left[\hat{K},\,\hat{K}_Z\right]&\hat{K}_{ZZ}=\frac{1}{2}\left[\hat{K}_Z,\,\hat{Z}\right]+\kappa\,\hat{Z}_{KK}\\
\hat{K}_{XY}=\frac{1}{2}\left[\hat{K}_X,\,\hat{Y}\right]&\hat{K}_{YZ}=\frac{1}{2}\left[\hat{K}_Y,\,\hat{Z}\right]&\hat{K}_{ZX}=\frac{1}{2}\left[\hat{K}_Z,\,\hat{X}\right]\\
\hat{L}_Z=\frac{1}{2}\left[\hat{K}_X,\,\hat{K}_Y\right]&\hat{L}_X=\frac{1}{2}\left[\hat{K}_Y,\,\hat{K}_Z\right]&\hat{L}_Y=\frac{1}{2}\left[\hat{K}_Z,\,\hat{K}_X\right]
\end{array}
\end{equation}

\noindent Thus {\bf(a)} is proven, but we note in passing that an alternative linearly independent set is as above but with $\hat{L}_X,\,\hat{L}_Y,\,\hat{L}_Z$ replaced by $\hat{X}_KK,\,\hat{Y}_{KK},\,\hat{Z}_{KK}$. However, in this case, the determinant stated is $-\frac{1}{2}\left(\frac{\gamma_e-\gamma_n}{\gamma_e+\gamma_n}\right)^3$ and vanishes when $\gamma_n=\gamma_e$. So, although we use this latter set for our realisation below (it is simpler), it cannot be used for coupled spin systems when $\gamma_e\approx\gamma_n$.

\noindent {\bf(b)} Without loss of generality, one can align the co-ordinate axes to the field to show by simple computation that {\it e.g.} $\{\hat{Z},\,\hat{K},\,\hat{K}_Z,\,\hat{K}_{ZZ},\,\hat{Z}_{KK}\}$ is the claimed smallest algebra, that $\{\hat{K}_Z,\,\hat{K}_{ZZ},\,\hat{Z}_{KK}\}$ is an ideal of this algebra and that the ideal is isomorphic to $\mathfrak{su}(2)$.  For the three separate ideals spanned by the $\hat{K}_X\,\cdots$, $\hat{K}_Y\,\cdots$ and $\hat{K}_Z\,\cdots$  matrices, it is readily shown that there are {\it two} seperate central elements for the exponentials of these ideals, so that the fundamental groups of these exponentials are the same as those of $SU(2)$, not $SO(3)$. Lastly, $[\hat{Z},\,\hat{K}]$ lies inside the ideal, so therefore the corresponding cosets commute in in the factor algebra and the quotient group must be a commutative Lie group. $e^{\hat{K}\,t}$ is, as we have seen, periodic, thus the one parameter group defined by $\hat{K}$ is $U(1)$. The one parameter group defined by $\hat{Z}$ is $U(1)$ if $\gamma_e / \gamma_n$ or $\gamma_n/\gamma_e$ is a rational number, otherwise it is the noncompact group $(\mathbb{R},\,+)$, analogous to the irrational slope one parameter subgroup on a 2-torus.

\noindent {\bf(c)} Readily proven as for {\bf (b)}. In this case, it is $\{\hat{K}_Z,\,\hat{K}_{ZZ},\,\hat{Z}\}$ is the ideal exponentiating to $SU(2)$.

\noindent {\bf (d)} Without loss of generalness, we can assume $\hat{X},\,\hat{Y}$ as the two directions: if there are two linearly independent directions, then these span a plane so that we can choose orthogonal vectors spanning the plane. Then we derive $\hat{K}_{XY},\,\hat{K}_X,\,\hat{K}_Y$ from $\hat{K},\,\hat{X},\,\hat{Y}$ as in \eqref{FullLieAlgebraBrackets} and then form:

\begin{equation}
\label{TwoGeneratorLieAlgebra}
\begin{array}{cl}
&[\hat{K}_{XY},\,X]+\kappa\,\left(\frac{\gamma_n-\gamma_e}{\gamma_n+\gamma_e}\right)^2\,[\hat{K}_x,\,\hat{K}_Y]+\frac{\kappa^2\,(\gamma_n-\gamma_e)}{\gamma_n+\gamma_e}\,\hat{K}+2\,\kappa^2\,\hat{K}_X+2\,\frac{\gamma_e\,\gamma_n}{(\gamma_e+\gamma_n)^2}\,K_Y\\=&\frac{\kappa\,(\gamma_n-\gamma_e)}{\gamma_n+\gamma_e}\,\hat{Z}
\end{array}
\end{equation}

\noindent so that, as long as $\gamma_e\neq\gamma_n$, a finite sequence of linear and bracket operations yields the vector $\hat{Z}$, whence the generation of the whole algebra proceeds exactly as in {\bf (a)}.

\noindent {\bf (e)} As can be checked by straighforward but tedious calculation, the smallest algebra generated by $\hat{K},\,\hat{X},\,\hat{Y}$ when $\gamma_e=\gamma_n$ is \linebreak$\mathrm{span}_\mathbb{R}\left(\left\{\hat{K},\,\hat{X},\,\hat{Y},\,\hat{K}_X,\,\hat{K}_Y,\hat{K}_{XX},\,\hat{K}_{YY},\,\hat{K}_{XY},\,\hat{L}_Z\right\}\right)$. The following linear combinations:

\begin{equation}
\label{GellMannBasis}
\begin{array}{ccc}
\boldsymbol{\hat{\lambda}}_1=\hat{K}_{XX}-\hat{K}_{YY}&\boldsymbol{\hat{\lambda}}_2=2(\hat{K}_{XY}-\kappa\,\hat{X}+\kappa^2\,\hat{K})&\boldsymbol{\hat{\lambda}}_3=2\,\hat{L}_Z\\\boldsymbol{\hat{\lambda}}_4=\frac{\hat{X}-\kappa\,\hat{K}+\hat{K}_Y}{\sqrt{2}}&
\hat{\lambda_5}=\frac{\hat{Y}-\kappa\,\hat{K}-\hat{K}_X}{\sqrt{2}}\\\boldsymbol{\hat{\lambda}}_6=\frac{\hat{X}-\kappa\,\hat{K}-\hat{K}_Y}{\sqrt{2}}&\boldsymbol{\hat{\lambda}}_7=-\frac{\hat{Y}-\kappa\,\hat{K}+\hat{K}_X}{\sqrt{2}}&\boldsymbol{\hat{\lambda}}_8=\frac{\hat{K}_{XX}+\hat{K}_{YY}}{\sqrt{3}}\\
\end{array}
\end{equation}

\noindent span an eight-dimensional ideal, thus exponentiating to a normal subgroup and the factor group is the one parameter group $e^{\hat{K}\,t}$, which is periodic, thus isomorphic to $U(1)$. The eight vectors in \eqref{GellMannBasis}  have the same structure constants as those of the Gell-Mann matrices (divided by $i$, since the physicist's convention is work with the algebra of Hermitian matrices and to form exponentials as $e^{i\,H}$) and thus of $\mathfrak{su}(3)$. The basis in \eqref{GellMannBasis} begets the diagonal Killing form $\mathrm{Tr}(\mathrm{ad}(\boldsymbol{\hat{\lambda}}_j)\,\mathrm{ad}(\boldsymbol{\hat{\lambda}}_k))=-12\,\delta_{j\,k}$ with $\boldsymbol{\hat{\lambda}}_1,\,\boldsymbol{\hat{\lambda}}_2,\,\boldsymbol{\hat{\lambda}}_3$ spanning a copy of $\mathfrak{su}(2)$, $\boldsymbol{\hat{\lambda}}_8$ commuting with any of these three,  $\boldsymbol{\hat{\lambda}}_1,\,\boldsymbol{\hat{\lambda}}_2,\,\boldsymbol{\hat{\lambda}}_3,\,\boldsymbol{\hat{\lambda}}_8$ spanning the Lie subalgebra in a Cartan pair of a Cartan decomposition and the other four matrices spanning the orthogonal complement in the Cartan pair. Now we must decide whether the Lie algebra spanned by \eqref{GellMannBasis} exponentiates to the whole of $SU(3)$ or a discrete projection thereof ($PSU(3)$).  For $SU(3)$, the group's three-element center $\left\{\left.\exp\left(\frac{2\,k\,\pi}{3}\,i\right)\,\mathbf{1}\right|\;k\in\mathbb{Z}\right\}$ is contained within the one parameter subgroup corresponding to $\left\{\left.\exp\left((3\,\boldsymbol{\hat{\lambda}}_3+2\,\sqrt{3}\,\boldsymbol{\hat{\lambda}}_3)\,\tau\right)\right|\;\tau\in\mathbb{R}\right\}$. For the matrices in \eqref{GellMannBasis}, it is also readily shown that there are three distinct, central elements of the form $\exp\left((3\,\boldsymbol{\hat{\lambda}}_3+2\,\sqrt{3}\,\boldsymbol{\hat{\lambda}}_3)\,\frac{2\,k\,\pi}{3}\right);\,k\in\mathbb{Z}$, so the fundamental group and Lie algebra of the exponential are the same as that of $SU(3)$ and the exponential is therefore isomorphic to $SU(3)$, not a projection ($PSU(3)$) thereof.\qedhere

\end{proof}

\noindent {\bf Practical consequences for experimenters} The practical outcomes of Lemma \ref{FullLieAlgebraLemma}, given Theorem \ref{GeneratedLieSubgroupTheorem} and Algorithm \ref{LieAlgebraBasisFindingAlgorithm}, are {\bf (i)} that an experimental kit which can impart variable magnetic fields of either sign in three, linearly independent directions to a coupled spin system can always impart any time evolution operator from $SU(4)$ to the two coupled quantum spin state,  {\bf (ii)} that a simplified experimental kit imparting magnetic fields in only two linearly independent directions can do the same (at the expense of more complicated magnetic field pulse sequences) as long as $\gamma_n\neq\gamma_e$, so that for coupled spin systems with similar or the same gyromagetic ratios for the two coupled particles, all three independent magnetic fields will be needed for good controllability of the time evolution operator.

Lastly, I now give an example of the algorithms' application for the coupled $^{31}P$ nucleus-electron system. We first need to calculate a basis for $\mathfrak{su}(4)$ with  Algorithm \ref{LieAlgebraBasisFindingAlgorithm}. To this end, let us define our ``unit'' magnetic field strength as $B_u=10{\rm mT}$ and our ``unit'' time interval as $\tau_u=100{\rm ns}$. Under these unit conditions, we define $\hat{K}_U = 5.8765\,\hat{K}$, $\hat{X}_U = B_0\,(\gamma_n+\gamma_e)\,\tau_u\,\hat{X}_0 = 27.98723\,\hat{X}_0$, $\hat{Y}_U = 27.98723\,\hat{Y}_0$ and $\hat{Z}_U=27.98723\,\hat{Z}_0$. By imparting a magnetic field (possibly none) $B_x,\,B_y,\,B_z\in\mathbb{R}$ units for a time interval $\tau$  units, we can now directly impart any time evolution operator of the form $\exp((B_x,\hat{X}_U + B_y,\hat{Y}_U+B_z\,\hat{Z}_U) + \hat{K}_U)\tau)$ for $\tau\geq 0$. As discussed in \S\ref{Introduction}, we can achieve a term of the form $-\hat{K}_U\,t$ for $t\geq 0$ as $+\hat{K}_U\,(n\,\tau_p - t)$  where $n\in\mathbb{N}$ and $\tau_p\approx 2.14$  in our $100{\rm ns}$ time units is the period of $e^{\hat{K}_U\,t}$. 

By examination of the steps of Theorem\ref{GeneratedLieSubgroupTheorem} and Algorithm \ref{LieAlgebraBasisFindingAlgorithm}, we can impart time evolution operators belonging to the one parameter subgroups $\{e^{\tau\,\hat{H}_j}|\;\tau\in\mathbb{R}\}$ corresponding to the following Lie algebra members $\hat{H}_j$:

\begin{equation}
\label{ControlAlgebra}
\begin{small}
\begin{array}{l}
\begin{array}{ll}
\hat{H}_1 = (b_1\,\hat{X}_U + \hat{K}_U)\,\tau_1&\hat{H}_2 =( b_2\,\hat{Y}_U + \hat{K}_U)\,\tau_2\\
\hat{H}_3 =( b_3\,\hat{Z}_U + \hat{K}_U)\,\tau_3&\hat{H}_4 = \tau_4\,\hat{K}_U\\\\
\hat{H}_5 =\mathrm{Ad}(\exp((b_5\,\hat{X}_U + \hat{K}_U)\,\tau_5))\,\hat{H}_4&\hat{H}_6 =\mathrm{Ad}(\exp((b_6\,\hat{X}_U + \hat{K}_U)\,\tau_6))\,\hat{H}_4\\
\hat{H}_7=\mathrm{Ad}(\exp(\hat{K}_U\,\tau_7))\,\hat{H}_1&\\\\
\hat{H}_8 =\mathrm{Ad}(\exp((b_8\,\hat{Y}_U + \hat{K}_U)\,\tau_8))\,\hat{H}_4&\hat{H}_9 =\mathrm{Ad}(\exp((b_9\,\hat{Y}_U + \hat{K}_U)\,\tau_9))\,\hat{H}_4\\
\hat{H}_{10}=\mathrm{Ad}(\exp(\hat{K}_U\,\tau_{10}))\,\hat{H}_2&\\\\
\hat{H}_{11} =\mathrm{Ad}(\exp((b_{11}\,\hat{X}_U + \hat{K}_U)\,\tau_{11}))\,\hat{H}_4&\hat{H}_{12} =\mathrm{Ad}(\exp((b_{12}\,\hat{X}_U + \hat{K}_U)\,\tau_{12}))\,\hat{H}_4
\end{array}\\\\
\hat{H}_{13} =\mathrm{Ad}(\exp((b_{13}\,\hat{X}_U + \hat{K}_U)\,\tau_{13}))\,\mathrm{Ad}(\exp((c_{13}\,\hat{Y}_U + \hat{K}_U)\,\varsigma_{13}))\,\hat{H}_4\\
\hat{H}_{14} =\mathrm{Ad}(\exp((b_{14}\,\hat{Y}_U + \hat{K}_U)\,\tau_{14}))\,\mathrm{Ad}(\exp((c_{14}\,\hat{Z}_U + \hat{K}_U)\,\varsigma_{14}))\,\hat{H}_4\\
\hat{H}_{15} =\mathrm{Ad}(\exp((b_{15}\,\hat{Y}_U + \hat{K}_U)\,\tau_{15}))\,\mathrm{Ad}(\exp((c_{15}\,\hat{Z}_U + \hat{K}_U)\,\varsigma_{15}))\,\hat{H}_4\\
\end{array}
\end{small}
\end{equation}

We can here use the second basis (with $\hat{X}_{KK},\,\hat{Y}_{KK},\,\hat{Z}_{KK}$)  discussed in the proof of Lemm \ref{FullLieAlgebraLemma}, part {\bf (a)}. We can understand why the forms in \eqref{ControlAlgebra} do indeed yield an independent basis by looking carefully at the proof of Theorem \ref{GeneratedLieSubgroupTheoremProof} and Algorithm \ref{LieAlgebraBasisFindingAlgorithm}. Firstly, $\hat{H}_1,\,\cdots,\,\hat{H}_4$ are linearly independent. The forms of $\hat{H}_5,\,\cdots,\,\hat{H}_6,\,\hat{H}_8,\,\hat{H}_9$, $\hat{H}_{11},\,\hat{H}_{12}$ are the forms in $\eqref{GeneratedLieSubgroupTheoremProof_3}$ and \eqref{GeneratedLieSubgroupTheoremProof_4} with $\sigma_X$ set to paths of the form $\exp(s\,\hat{H}_j)$ for $j\in1\cdots3$ and $Y$ set to $\sigma_Y$ set to $\exp(\tau\,\hat{K_u})$ with different values of $s$ in \eqref{GeneratedLieSubgroupTheoremProof_4}. Since we have seen that $\hat{X}_U$ and $\hat{K}_U$ generates a Lie algebra of dimension 5, there can only be two linearly independent values of the form of $\hat{H}_5$ and  $\hat{H}_6$ and one of the form $\hat{H}_7$ realised with different values of $s$in \eqref{GeneratedLieSubgroupTheoremProof_4}. Likewise for $\hat{H}_8,\,\hat{H}_{9},\,\hat{H}_{10}$. We realise two further linearly independent values  $\hat{H}_{11},\,\hat{H}_{12}$ likewise but when try to get a further independent value of the form $\mathrm{Ad}(\exp(\hat{K}_U\,\tau_7))\,\hat{H}_3$ we fail. This is because we have already included $\hat{K}_U$ in our list of linearly independent vectors, and then $\hat{Z}_{KK}$ is linearly depenent on all the others through the relationship $\hat{X}_{KK}+\hat{Y}_{KK}+\hat{Z}_{KK}=-2\,\hat{K}$. Lastly, the forms of $\hat{H}_{13},\,\hat{H}_{14},\,\hat{H}_{15}$ give us the last three linearly independent vectors because we have seen in Lemma \eqref{ControlAlgebra} that $\mathrm{ad}(\hat{X})\,\mathrm{ad}(\hat{Y})\,\hat{K}$,$\mathrm{ad}(\hat{Y})\,\mathrm{ad}(\hat{Z})\,\hat{K}$ and $\mathrm{ad}(\hat{Z})\,\mathrm{ad}(\hat{X})\,\hat{K}$ yield the last three members of the $\frak{su}(4)$ basis.

Now we must choose the numerical values of the constants in \eqref{ControlAlgebra}. This was done numerically by vectorising all the fifteen $\hat{H}_j$ (calculating their components with respect to the basis $\{\mathbf{u}_j\otimes \mathbf{u}_k|\;j,\,k\in 0\cdots 3\}$ with $\mathbf{u}_0=i\,\mathbf{1},\mathbf{u}_1=\mathbf{i},\,\mathbf{u}_2=\mathbf{j},\,\mathbf{u}_3=\mathbf{k}$), assembling these components into a $15\times 16$ real matrix and then calculating the singular values of this matrix. The condition number (ratio of maximum to minimum singular value) measures {it how} linearly independent the basis is and therefore how well it is likely to work in Algorithm \ref{CanonicalCoordinatesFindingAlgorithm}. That is, poor (large) condition numbers mean higher values of $n$ output by Algorithm \ref{CanonicalCoordinatesFindingAlgorithm} and therefore much longer preparation magnetic pulse sequences. So the constants of \eqref{ControlAlgebra} were optimised with a simple hill climbing procedure (repeatedly working through the list, optimising each constant individually) to find the constants yielding the minimum condition number for the basis. The following constants yielded a maximum singular value of $15.82$ and a minimum singular value of $1.70$, {\it i.e.} a condition number of $9.3$: $b_1 = 2.003$, $\tau_1 = 0.151$, $b_2 = 1.5155$, $\tau_2 = 0.176$, $b_3 = 1.958$, $\tau_3 = 0.118$, $\tau_4 = 1.109$, $b_5 = 0.3015$, $\tau_5 = 1.021$, $b_6 = 0.5195$, $\tau_6 = 0.910$, $\tau_7 = 0.215$, $b_8 = 0.1925$, $\tau_8 = 0.931$, $b_9 = 0.222$, $\tau_9 = 0.926$, $\tau_{10} = 0.2005$, $ b_{11} = -0.167$, $\tau_11 = 0.9825$, $b_{12} = 0.394$, $\tau_{12} = 0.9255$, $b_{13} = 0.198$, $\tau_{13} = 1.017$, $c_{13} = 0.178$, $\varsigma_{13} = 0.9855$, $b_{14} = 0.344$, $\tau_{14} = 1.000$, $c_{14} = 0.166$, $\varsigma_{14} = 0.9845$, $b_{15} = 0.257$$, $$\tau_{15} = 0.900$, $c_{15} = 0.190$, $\varsigma_{15} = 1.377$. 

Let us now use this basis to realise the unitary $SU(4)$ matrix $\boldsymbol{\gamma}=\mathbf{j}\otimes\mathbf{i}$. This has a physical meaning of the change that we should see if we could impart a magnetic field in the $y$ direction on the electron, a magnetic field of the same magnitude but in the $x$ direction on the nucleus and allow both of these spin states to precess, uncoupled and independently, a one-fourth cycle about their respective magnetic fields. Clearly there is no simple, obvious way of doing this with the fields available to the experimenter. The matrix logarithm of $\boldsymbol{\gamma}$ is $\log \boldsymbol{\gamma} = -i\,\frac{\pi}{2}\,(\mathbf{i}\otimes\mathbf{j}+\mathbf{k}\otimes\mathbf{k})$ and its approximate components in the basis of \eqref{ControlAlgebra} are $x_1=-0.287893$, $x_2=-0.41226$, $x_3=0.178931$, $x_4=-0.846392$, $x_5=0.248348$, $x_6=0.215918$, $x_7=0.163212$, $x_8=0.77681$, $x_9=0.143116$, $x_{10}=0.204211$, $x_{11}=0.19128$, $x_{12}=0.219224$, $x_{13}=0.517963$, $x_{14}=-0.363032$, $x_{15}=-0.269563$.

We now use a Runge-Kutta discretised approximation of Algorithm \ref{CanonicalCoordinatesFindingAlgorithm}. One thousand steps ($\Delta t=0.001$) is found to be extremely accurate. When we try to find the canonical co-ordinates of $\boldsymbol{\gamma}$ directly by the first for loop, the determinant of the Wei-Norman matrix vanishes at about step number 166, {\it i.e.} at about one sixth the way through the loop, signalling that roughly $\exp\left(\frac{1}{7}\,\log \boldsymbol{\gamma}\right)$ is the furthest matrix from the identity that the Wei Norman equations can find second kind canonical co-ordinates for. Therefore, the second for loop in Algorithm  \ref{CanonicalCoordinatesFindingAlgorithm}, we set $n=7$ and find the canonical co-ordinates of $\exp\left(\frac{1}{7}\,\log \boldsymbol{\gamma}\right)$. Now the Wei-Norman equations can be solved and yield the following canonical co-ordinates of the second kind:

\begin{equation}
\label{NumericalResult}
\begin{array}{l}
\{h_j\}_{j=1}^{15}\approx \begin{small}\begin{array}{ccccccc}\{&-0.0724497&-0.0480297&0.0459191&-0.0680989&0.0333669&\\&0.0469188&0.0328527&-0.125384&
0.203289&0.0477131&\\&0.019957&0.0193781&0.0793817&-0.0672283&-0.0156309&\}\end{array}\end{small}\\\\
\exp\left(\frac{1}{7}\,\log \boldsymbol{\gamma}\right) = \prod\limits_{j=1}^{15}\,e^{h_j\,\hat{H}_j}\approx \boldsymbol{\gamma}_{\frac{1}{7}}\stackrel{def}{=}\left(\begin{array}{cccc}\gamma_{1 1}&0&0&\gamma_{1 4}\\0&\gamma_{2 2}&-\gamma_{1 4}^*&0\\0&\gamma_{1 4}^*&\gamma_{2 2}&0\\-\gamma_{1 4}&0&0&\gamma_{1 1}\end{array}\right)\\\\
\gamma_{1 1}= 0.950484 +0.216942 i;\;\gamma_{1 4}=-0.216942-0.0495156 i\\\gamma_{2 2}=0.950484 -0.216942 i
\end{array}
\end{equation}

\noindent and the rms error of the matix elements $\sqrt{\frac{1}{16}\mathrm{tr}((\boldsymbol{\gamma}_{\frac{1}{7}}^7-\boldsymbol{\gamma})^\dagger\,(\boldsymbol{\gamma}_{\frac{1}{7}}^7-\boldsymbol{\gamma})}$ is found to be $6\times 10^{-12}$, where $\boldsymbol{\gamma}_{\frac{1}{7}}$ is the numerical value defined by the numerical second kind canonical co-ordinates in \eqref{NumericalResult} and $\boldsymbol{\gamma}_{\frac{1}{7}}^7$ is the numerical result of Algorithm \ref{CanonicalCoordinatesFindingAlgorithm}, {\it i.e.} the numerical realisation of our design goal $\boldsymbol{\gamma}=\mathbf{j}\otimes\mathbf{i}$. The lengths of the magnetic pulses found here are of the order of tens to hundreds of nanoseconds, and the speed of pulse manipulation is thus well below the $1{\rm GHz}$ upper frequency limits of the contemporary magnetic field control speeds cited in \cite{KuzmakTkachuk}. The peak magnetic fields strengths used here are of the order of $2{\rm mT}$ ({\it e.g.} the magnetic field in $h_1\,(b_1\,\hat{X}_U +\hat{K}_U)$, which is $2.003\times0.287893\approx0.58$, corresponds to $5.8{\rm mT}$, and the constants $h_j$ in the canonical co-ordinates are between $0.1$ and $0.3$); these are big but are quite in keeping with the peak manipulation fields available in today's technology. Lastly, one Wei-Norman cycle, {\it i.e.} preparation sequence of the form $e^{h_1\,\hat{H}_1}\,e^{h_2\,\hat{H}_2}\,\cdots e^{h_{15}\,\hat{H}_{15}}$ takes forty six magnetic pulsing/ idling stages. The whole preparation therefore takes 322 pulsing / idling  stages, each of the order of ten nanoseconds. Assuming this example to be typical, the whole preparation sequence takes of the order of 5 microseconds, well smaller than the two second coherence time cited in \cite{KuzmakTkachuk}.

\section{Conclusion}

In a general Lie group $\mathfrak{G}$ setting, we have discussed the set generated by Lie group members of the form $e^{\tau_j\,\hat{X}_j}$ where $\hat{X}_j\in \mathfrak{g}$ do not span a Lie algebra and given a semiconstructive proof of the fact that they generate, with products of a {\it finite} number of terms, the (possible non-topologically embedded) Lie subgroup generated by $\exp(\mathfrak{h})$, where $\mathfrak{h}$ is the smallest Lie algebra containing the $\hat{X}_j$. This proof becomes constructive when the Lie algebra exponentiates to the whole Lie group, as happens in a compact group, for example, and we have given algorithms and proofs of their convergence in such a case. In the case of two qubit coupled spin $\frac{1}{2}$ quantum states, these proofs show that the ability to pulse three linearly independent magnetic fields at any field strength in an interval around nought for any positive time allows the experimenter to impart any unitary transformation on the quantum state space in a finite sequence of these operations. Indeed, when the gyromagnetic ratios for the coupled spins are strongly unalike, two linearly independent magnetic field directions will suffice. A numerical version of the algorithms was run for an example transformation and found to be highly accurate. 

\bibliographystyle{siam}
\bibliography{LieTheoreticStateManipulation} 

\end{document}